\documentclass[lettersize,journal]{IEEEtran}
\usepackage[T1]{fontenc}
\usepackage{mathrsfs}
\usepackage{bm}
\usepackage{amsmath}
\usepackage{verbatim}
\usepackage{amsthm}
\usepackage{amssymb}
\usepackage{afterpage}
\usepackage{graphicx}
 \PassOptionsToPackage{normalem}{ulem}
\usepackage{hyperref}
\usepackage{bbm}
\makeatletter

\usepackage{tabularx}

\DeclareMathOperator{\esssup}{ess\: sup}
\usepackage[switch]{lineno}
\usepackage[named]{algo}
\usepackage{algpseudocode}
\usepackage{algorithm}
\usepackage{subcaption}

\newcommand{\wadd}{\textrm{WADD}}
\newcommand{\arl}{\textrm{ARL}}

\theoremstyle{plain}
\newtheorem{theorem}{\protect\theoremname}
  \theoremstyle{plain}
  
  \theoremstyle{plain}
  
  \theoremstyle{plain}
   \newtheorem{lemma}{\protect\lemmaname}
  
  \newtheorem{corollary}{Corollary}
\theoremstyle{remark}
\newtheorem{remark}{Remark}

\usepackage{graphicx,psfrag,cite}
\usepackage{enumitem}
\usepackage[table]{xcolor}
\usepackage{relsize}
\usepackage{amsmath}
\DeclareMathOperator*{\argmax}{arg\,max}

\setkeys{Gin}{width=1.0\columnwidth}

\makeatother
  \providecommand{\definitionname}{Definition}
  \providecommand{\lemmaname}{Lemma}
  \providecommand{\propositionname}{Proposition}
  
\providecommand{\theoremname}{Theorem}
\providecommand{\conjecturename}{Conjecture}

\begin{document}

\title{Quickest change detection for UAV-based sensing 

 \thanks{Saqib Abbas Baba is with the Department of Electrical Engineering, Indian Institute of Technology (IIT), Delhi. Email: saqib.abbas@ee.iitd.ac.in}
  \thanks{ Anurag Kumar is with the Department of Electrical Communication Engineering, IISc Bangalore. Email: anurag@iisc.ac.in }
  \thanks{Arpan Chattopadhyay is with the Department of Electrical Engineering and the Bharti School of Telecom Technology and Management, IIT Delhi. Email: arpanc@ee.iitd.ac.in }
}

\author{
 Saqib Abbas
  \hspace{0.1cm} Anurag Kumar
  \hspace{0.1cm} Arpan Chattopadhyay \vspace {-0.4in}
}

\maketitle

\begin{abstract}
This paper addresses the problem of quickest change detection (QCD) at two spatially separated locations monitored by a single unmanned aerial vehicle (UAV)  equipped with a sensor. At any location, the UAV observes i.i.d. data sequentially in discrete time instants. The distribution of the observation data changes at some unknown, arbitrary time and the UAV has to detect this change in the shortest possible time. Change can occur at most at one location over the entire infinite time horizon.   The UAV switches between these two  locations in order to quickly detect the change. To this end, we propose   Location Switching and Change Detection (LS-CD) algorithm which uses  a repeated one-sided  sequential probability ratio test (SPRT) based mechanism for observation-driven location switching and change detection. The primary goal is to minimize the worst-case average detection delay (WADD) while meeting constraints on the average run length to false alarm (ARL2FA) and the UAV’s time-averaged energy consumption. We provide a rigorous theoretical analysis of the algorithm’s performance by using theory of random walk. Specifically, we derive tight   upper and lower  bounds to its ARL2FA and a tight upper bound to  its WADD.   In the special case of a symmetrical setting, our analysis leads to a new asymptotic upper bound to the ARL2FA of the standard CUSUM algorithm,  a novel contribution not available in the literature, to our knowledge. Numerical simulations  demonstrate the efficacy of LS-CD.
\end{abstract}

\section{Introduction}\label{Introduction}

Many critical applications such as surveillance, intrusion detection,  environmental monitoring, predictive maintenance and security systems,  often require  detection of  changes or anomalies across multiple locations.   Traditional literature on quickest change detection (QCD \cite{poor2008quickest}) seeks to minimize detection delay for such applications under constraints on false alarm rates, albeit primarily for a single data stream. Seminal contributions by Shiryaev \cite{shiryaev1963optimum} on Bayesian QCD  and Lorden \cite{lorden1971procedures} on non-Bayesian QCD laid the theoretical foundation for detection procedures based on optimal stopping rules. In particular,  Lorden’s min-max formulation \cite{lorden1971procedures} and its  later extension by Moustakides \cite{moustakides1986optimal} established asymptotic and non-asymptotic optimality of the CUSUM (Cumulative Sum) algorithm (originally introduced by Page \cite{page1954continuous}) for the problem of   minimizing WADD under a constraint on ARL2FA for a single data stream with i.i.d. observations in pre and post change regimes. Subsequent work by  Lai \cite{lai1998information} further established the optimality of CUSUM in a number of settings with non-i.i.d. observation sequence. On the other hand, many prior studies on multi-sensor monitoring \cite{lai2010sequential, veeravalli2001decentralized, tartakovsky2004change, mei2010efficient, xie2013sequential, fellouris2017multistream, tartakovsky2014sequential, banerjee2015data,rovatsos2021quickest,wang2015large,xie2019asynchronous,raghavan2010quickest,zou2019quickest,hadjiliadis2009one} primarily focused on monitoring a single process by multiple sensors and often communicating them over error-free communication channels to a remote decision making node, which effectively reduces the problem to QCD performed by a single sensor monitoring a single observation stream. Recently, the paper \cite{xu2021optimum} has considered multi-stream change-point detection where the sensor may sample only one data stream at a time and can switch to any other  data stream in the next time instant without any cost and delay. The authors propose a simple greedy sampling policy combined with a CUSUM test and establish its asymptotic optimality. At each time step, the algorithm uses a myopic sampling policy and updates a local change statistic, similar to a CUSUM score. If a change is detected, a global alarm is triggered; otherwise, the procedure continues sampling other streams based on the myopic sampling policy. In scenarios where only a subset of sensors can be observed at each time step, the paper \cite{zhang2019partially} proposed an online learning framework  to detect change-points from partial observations-  it formulated the sensor subset selection problem as a multi-armed bandit problem.

\begin{figure}[t!]
    \centering
    \includegraphics[height=4.5cm, width=1\linewidth]{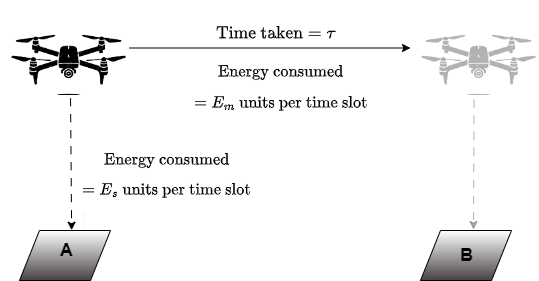}
    \vspace{-4mm}
    \caption{System setup}
    \label{fig:sys_setup}
\end{figure}

However, monitoring multiple processes across multiple locations by using a mobile sensor such as a UAV introduces additional challenges such as switching delay of UAV,  energy expenditure for the UAV's movement and hovering, optimal location switching algorithm design for the UAV, etc. In this paper, we seek to address these challenges for the scenario in which a single UAV monitors two independent processes at two spatially separated locations by dynamically switching between them in order to detect change in observation distribution  at either location. To this end, we propose the LS-CD algorithm  in which the UAV dynamically updates a detection statistic at each visited location in order to decide whether to continue sensing there, or whether to raise an alarm, or whether to switch location. We analyze ARL2FA and WADD performance of this algorithm, and also numerically demonstrate its energy consumption benefits.

The main contributions of this paper are summarized as follows:
\begin{enumerate}
    \item We propose the LS-CD  algorithm to minimize WADD under ARL2FA and energy consumption rate constraints for the UAV. 
    \item We derive tight upper and lower  bounds on the ARL2FA of LS-CD. For the special case with zero switching time of the UAV and identical pre and post change distributions across both locations, the upper bound specializes to a novel asymptotic upper bound to the ARL2FA of standard CUSUM algorithm under i.i.d. observation model in pre and post change regimes- this upper bound is only a constant factor away from the standard lower bound to the ARL2FA of CUSUM. ARL2FA analysis of LS-CD involves significant mathematical challenges that could be handled efficiently by invoking theory from random walk literature.
    \item We also derive a  tight upper bound for the WADD achieved by LS-CD. One additional significant challenge in this derivation was finding the minimal set of ``states" of the UAV at the change instant, that could give rise to worst case detection delay.
\end{enumerate}

Though the LS-CD algorithm is presented in this work for a two-location scenario, it naturally generalizes to settings with multiple locations. In practice, a UAV can employ a round-robin or cyclic switching strategy to decide the next location to be monitored, while preserving the structure of the local detection and switching rule at each  location. Our analyses can be easily  extended to this   general setting.

The remainder of the paper is organized as follows. Section~\ref{section:background} gives a brief introduction to the classical non-Bayesian QCD framework. Section~\ref{section:system-model} then describes the system model, the problem formulation and  the proposed LS-CD algorithm.  Section~\ref{section:Performance analysis} derives performance bounds for ARL2FA and WADD of LS-CD.  Section~\ref{sec:numerical_results} presents the numerical results.  Finally, conclusions are made in  Section~\ref{sec:conclusion}. All proofs are provided in the appendices.

\section{Background}
\label{section:background}

Quickest change detection (QCD) problems typically involve a sequence of observations $\{Y_t\}_{t\ge1}$ whose distribution changes at some unknown time $\nu$. Before the change occurs ($t < \nu$), the observations follow a known pre-change distribution $f$ and after the change ($t \geq \nu$), the observations follow a different post-change distribution $g$, i.e.,
\begin{align*}
    Y_t \sim \begin{cases}
        f \quad t < \nu \\
         g \quad t \geq \nu
    \end{cases} 
\end{align*}
A decision maker collects the observations $\{Y_t\}_{t\ge1}$ sequentially and decides to stop at a time $T$ and declare that a change has occurred.
The objective in QCD is to design a stopping time $T$ that detects the change as quickly as possible while controlling false alarms. 
The worst case detection delay is defined as~\cite{lorden1971procedures}
\[
  \wadd = \sup_{\nu\ge1} \esssup \mathbb{E}_\nu\bigl[(T-\nu)^+ \bigm| Y_1,\dots,Y_{\nu-1}\bigr],
\]
where $\mathbb{E}_\nu[\cdot]$ denotes the expectation under the probability law induced by the change occurring at time $\nu$. The goal is to minimize $\wadd$ subject to a constraint on the expected run length to false alarm:
\begin{equation} \label{eq:qcd_problem}
    \min_T \ \wadd, \quad \text{s.t. } \mathbb{E}_{\infty}[T] \geq r.
\end{equation}
Here, the ARL2FA, $\mathbb{E}_{\infty}[T]$ is the expected stopping time under the no-change regime (i.e., $\nu=\infty$), ensuring that false alarms do not occur too frequently.

Seminal results by Shiryaev~\cite{shiryaev1963optimum} and Lorden~\cite{lorden1971procedures} established the theoretical foundation for such detection procedures. Building on these, Page's CUSUM procedure~\cite{page1954continuous} was shown by Lorden to have strong optimality properties, and Moustakides~\cite{moustakides1986optimal} provided a rigorous proof that CUSUM minimizes Lorden's worst-case detection delay while respecting an $\mathrm{ARL}$ constraint.

The CUSUM statistic is updated at each time step $t$, and an alarm is raised when it exceeds a threshold $\gamma$. The worst-case average detection delay scales as $\frac{\gamma}{D(f\mid\mid g)}$, where $D(f\mid\mid g)$ denotes the Kullback--Leibler (KL) divergence between the pre-change and post-change distributions.

\section{System Model}\label{section:system-model}
We now extend the classical QCD framework to a scenario in which a single UAV moves between two spatially separated locations, $\mathcal{L} = \{A,B\}$. Time is discrete and indexed by $t=1,2,3,\cdots$. The time taken by the UAV to move from one location to another is $\tau$ time slots (deterministic), and requires $E_m$ amount of energy per slot (also deterministic). The energy spent by the UAV in a single time slot while it is hovering over and sensing at a location is given by $E_{s}$ (see Figure~\ref{fig:sys_setup}).

At time~$t$, if the UAV is monitoring location~$l \in \mathcal{L}$, a measurement $Y_{l,t} \in \mathbb{R}^{s \times 1}$ is obtained. We assume that the measurements $\{Y_{l,t}\}_{t \geq 1}$ are independent across the two locations. At any location~$l$, the change event of our interest occurs at an unknown, arbitrary time $\nu_l$. \textit{We assume that the change can occur in at most one location over the entire time horizon.}

Each location $l \in \mathcal{L}$ is associated with two hypotheses:
\begin{eqnarray*}
 && H_0^{l}: Y_{l,t} \sim f_{l}(\cdot) \text{ i.i.d. for $t < \nu_l$}, \\
&& H_1^{l}: Y_{l,t} \sim g_{l}(\cdot) \text{ i.i.d. for $t \geq \nu_l$}.
\end{eqnarray*}
Here $\nu_l$ is the unknown change point at location~$l$, and $f_{l}(\cdot)$ and $g_{l}(\cdot)$ are the pre-change and post-change distributions of the observations, respectively, assumed to be known to the decision making algorithm in the UAV. Note that, $\nu_l=\infty$ implies that no change ever occurs at location $l$.

Since the UAV physically moves between locations $A$ and $B$, only one location can be observed at a time. The UAV follows a switching rule $\pi$ that determines when to switch to the other location, based on observations. Let $\tau_l$ represent a typical sojourn time of the UAV at location~$l$, and denote by $\overline{\tau_l}$ its mean and by $\overline{\tau_l^2}$ its second moment when no change occurs at location $l$ (i.e., $\nu_l=\infty$). These moments are dependent on the switching policy employed by the UAV and the observation statistics at location $l$ only. Let $T'_A$ and $T'_B$ denote the stopping times at which the UAV declares a change at locations $A$ and $B$, respectively. The UAV must balance two key detection metrics: the \emph{worst-case average detection delay} ($\wadd_l$) and the \emph{average run length to false alarm} ($\arl_l$). Specifically, we use Lorden’s min--max approach~\cite{lorden1971procedures} for each $l$ and define:
\begin{align}\label{eqn:WADD_l_modified_lorden}
    \wadd_l = 
    \sup_{\nu_l\ge1}
    \esssup \mathbb{E}_{\nu_l}^{l}\bigl[(T'_l-\nu_l)^+ \mid \mathcal{I}_l^{\nu_l-1}\bigr]
\end{align}

\begin{figure}[t!]
\centering
    \includegraphics[width=\columnwidth]{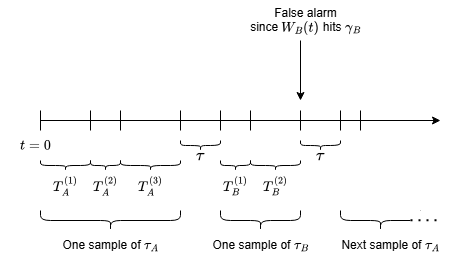}
    \vspace{-2mm}
    \caption{Illustration of the location switching procedure under the LS-CD algorithm with $n_A = n_B = 3$ and $\nu_A = \nu_B = \infty$. The intervals $T_A^{(1)}, T_A^{(2)}, T_A^{(3)}$ are i.i.d. samples from the distribution of $T_A$ under no change.}
    \label{fig:sprt_stop_time}
\end{figure}
Here, $\mathbb{E}_{\nu_l}^l$ denotes expectation assuming a change occurs at location~$l$ at time $\nu_l$. Unlike the standard Lorden’s formulation, where all observations are continuously available, the history $\mathcal{I}_l^{\nu_l-1}$ represents the collection of measurements the UAV has gathered at location $l$ upto time $\nu_l-1$, which is not necessarily complete due to the UAV’s intermittent presence at location $l$.

We define $\arl_l$ as the expected time till a false alarm is raised at either location, given that the UAV starts hovering at location $l$ at $t=0$ with $\nu_A=\nu_B=\infty$, and the QCD algorithm starts running at $t=0$.

By renewal reward theorem\cite{ross1995stochastic}, the average energy consumption in the pre-change regime (i.e., $\nu_A=\nu_B=\infty$) by the UAV per unit time step is \(\frac{2\tau E_{m}+ (\bar{\tau}_A+\bar{\tau}_{B}) E_s}{(\bar{\tau}_{A}+\bar{\tau}_{B}+2\tau)}\), which needs to be kept below a threshold $\bar{E}$.

Thus, the optimization problem is formulated as follows:
\begin{align}
&& \min_{\{T'_A,T'_{B},\pi\}} \max \{\wadd_{A}, \wadd_{B}\} \nonumber\\
\textbf{s.t. } &&  \arl_A \geq r_A, \arl_{B} \geq r_{B}, \nonumber\\
&& \frac{2\tau E_{m}+ (\bar{\tau}_A+\bar{\tau}_{B}) E_s}{(\bar{\tau}_{A}+\bar{\tau}_{B}+2\tau)} \leq \bar{E}\tag{CP} \label{opt:CP}
\end{align}
Note that, solving \eqref{opt:CP} is challenging due to the UAV's complex dynamics and limited observation capability, and computationally hard due to its highly non-convex nature.

Under classical Lorden’s formulation, CUSUM continuously observes data, and raises an alarm when a cumulative statistic exceeds a threshold. In our setting, the UAV can miss a change at location $A$ while it is observing $B$, or while in transit. Thus, the standard CUSUM’s assumption of uninterrupted observation does not apply. To handle this challenge, we propose the \emph{location switching--change detection (LS-CD)} algorithm as a heuristic approach for approximate solutions to \eqref{opt:CP}. LS-CD is motivated by the observation that CUSUM involves repeated one sided SPRTs. The LS-CD algorithm adaptively updates local CUSUM statistics at each visited location and decides when to switch based on the current statistics and a user-defined policy. Although LS-CD may incur some additional delay relative to a purely continuous CUSUM in a static scenario, it maintains an order-optimal alignment with classical CUSUM properties under Lorden’s criterion, under energy constraints. The subsequent sections provide a detailed description of LS-CD and theoretical bounds for its performance in terms of both $\arl_l$ and $\wadd_l$.

\subsection{The Location Switching and Change Detection Algorithm: LS-CD }\label{subsection:Algorithm}

At each monitored location, the detection process consists of repeated one-sided sequential probability ratio tests (SPRTs). Each cycle of the SPRT accumulates evidence for a change until either (i) the decision statistic exceeds a threshold, triggering a detection, or (ii) the statistic resets to zero, indicating no change. Since the UAV intermittently switches between locations, it can only update the statistic at the currently monitored location.

At time \(t\), if the UAV is monitoring location~\(l\), the following statistic is updated:
\begin{eqnarray}\label{eqn:CUSUM_statistic}
    W_l(t+1)=\max \Big\{W_l(t) + \log \left(\frac{g_l(Y_{l,t})}{f_l(Y_{l,t})}\right),0 \Big\} 
\end{eqnarray}

The Sequential Probability Ratio Test (SPRT) applied at location $l$ is defined as
\begin{eqnarray}\label{eqn:stopping_time}
    T_l = \inf \{t \geq 1: W_l(t) \notin (0,\gamma_l)\}, \quad W_l(0)=0,
\end{eqnarray}
assuming that the one-sided SPRT starts at $t=0$ at location $l$.
While the UAV monitors location $A$, the statistic \( W_B(t) \) remains unchanged due to the lack of observations at $B$, and vice versa.

\subsubsection*{Location Switching Rule}
The UAV monitors a location \( l \) until \( W_l(t) \) hits $0$ for \( n_l \) consecutive SPRT cycles.  Each reset to zero indicates that the pre-change hypothesis is likely to hold. Our rule asserts that if this return to $0$ happens in $n_l$ consecutive cycles, the UAV switches to the other location to explore potential changes there  (see figure~\ref{fig:sprt_stop_time}).

Upon switching, the statistic \( W_l(t) \) at the previous location is reset to zero, and the monitoring process continues at the new location using the same procedure.

\subsubsection*{Change Detection Rule}
At each monitored location \( l \), the CUSUM-like statistic \( W_l(t) \) accumulates evidence based on the observed data. When \( W_l(t) \) reaches the predefined threshold \( \gamma_l \), a change is declared. This threshold \( \gamma_l \) is set to control the false alarm rate, ensuring alarms only when there is strong evidence of a change.

\begin{remark}
While the modified Lorden’s WADD in  \eqref{eqn:WADD_l_modified_lorden}    incorporates partial observation history $\mathcal{I}_l^{\nu_l-1}$, the heuristic LS-CD  algorithm resets the statistic $W_l(t)$ to $0$ at the time of switching, and thus the history of observations gathered at location $l$ up to time $t$ is discarded by the algorithm.
\end{remark}

\section{Performance analysis}\label{section:Performance analysis}

In this section, we analyze the performance of the LS-CD algorithm introduced in Section~\ref{subsection:Algorithm}. Specifically, we establish bounds on ARL2FA and WADD at each location \( l \) under the LS-CD algorithm.

Let $\mathbb{P}_{\nu_l}^l$ and $\mathbb{E}_{\nu_l}^l$ represent the underlying probability law and expectation, respectively, when the change occurs at time $\nu_l$ at location $l$ with no change at the other location (i.e., $\nu_{l'}=\infty$ where $l' \doteq \{A,B\} \setminus l$). Similarly, $\mathbb{P}_{0}^l$ and $\mathbb{E}_{0}^l$ denote the probability law and expectation in the post-change regime at location $l$, corresponding to $\nu_l=0$ and $\nu_{l'}=\infty$. Finally, $\mathbb{P}_{\infty}$ and $\mathbb{E}_{\infty}$ correspond to the  probability law and expectation under no change at either location (i.e., $\nu_l=\nu_{l'}=\infty$). 

Also, let us define the log-likelihood ratio as \(Z_{l}(t) \doteq\log \left(\frac{g_l(Y_{l,t})}{f_l(Y_{l,t})}\right)\), but we simply denote it as $Z_l$ since $Z_l(t)$ is i.i.d. across $t$ under $H_0^l$ and $H_1^l$. Let $D(f_l \mid \mid g_l)$ be the Kullback--Leibler (KL) divergence between the pre-change distribution $f_l$ and post-change distribution $g_l$.

\subsection{One-Sided SPRT Analysis}
\label{subsection:one-sided-sprt}

We begin by examining a one-sided sequential probability ratio test (SPRT) at a single location \(l\) under the \emph{pre-change} regime. Let \(\gamma_l>0\) be the detection threshold, and let \(T_l\) be the (random) stopping time defined in a standard SPRT sense as in \eqref{eqn:stopping_time}). 

Define
\[
  \psi_{\infty,l}
  \;\doteq\;
  \mathbb{P}_{\infty}\!\Bigl(W_l(T_l) \,\ge \gamma_l\Bigr),
\]
i.e., the probability that the SPRT statistic crosses \(\gamma_l\) under the no-change probability  law \(\mathbb{P}_{\infty}\) for location $l$. Also define
\[
  \beta_l
  \;\doteq\;
  \mathbb{P}_{0}^l\!\Bigl(W_l(T_l)\,\le0\Bigr),
\]
which is the probability that \(W_l(t)\) hits \(0\) in the post-change regime at location $l$, i.e., $\nu_l=0,\nu_{l'}=\infty$. 

Let $\beta_l^{(w)} = \mathbb{P}_0^l \big(W_l(T_l) \le 0 \mid W_l(0) = w \big)$ be the post-change probability of the statistic falling below 0 given that it starts at some intermediate value $w$. Similarly, let, $\psi_{\infty,l}^{(w)} = \mathbb{P}_{\infty}\bigl(W_l(T_l)\ge\gamma_l \mid W_l(0)=w\bigr)$ denote the probability of crossing $\gamma_l$, starting from an initial statistic $W_l(0)=w$ under no change. Following \cite[Chapter~7]{ross1995stochastic}, the expression for $\psi_{\infty,l}^{(w)}$ is
\[
  \psi_{\infty,l}^{(w)}
  \;=\;
  \frac{\,1 - e^{-w}\,}{\,e^{(\gamma_l - w)} - e^{-w}\,}
  \;+\;
  O(1),
\]
where \(O(1)\) indicates a bounded correction due to overshoot. Consequently, $\psi_{\infty,l}$ can be decomposed as:
\begin{align*}
  &\psi_{\infty,l} \\
  &= \int_{0}^{\gamma_l} \psi_{\infty,l}^{(w)}
     d\mathbb{P}_{\infty}(Z_l = w) +
   \mathbb{P}_{\infty}\bigl(Z_l \ge\gamma_l\bigr)\\
  &= \int_{w=0}^{\gamma_l}
      \frac{1-e^{-w}}{e^{(\gamma_l - w)} - e^{-w}}
      d\mathbb{P}_{\infty}(Z_l=w) +
   \mathbb{P}_{\infty}\bigl(Z_l \ge\gamma_l\bigr)
   + O(1).
\end{align*}
where the second term accounts for the case where the statistic jumps above $\gamma_l$ in a single step.

Let \(\mathbb{E}_{\infty}\!\bigl(T_l^{(w)}\bigr)\) be the expected stopping time starting from $W_l(0)=w$, still under the pre-change regime. Using \cite[Chapter~7]{ross1995stochastic}, we can write:
\begin{align}\label{eqn:E_infty_T_l_w}
  \mathbb{E}_{\infty}\!\bigl(T_l^{(w)}\bigr) = 
  \frac{(\gamma_l - w)(1-e^{-w}) - w\,\bigl(e^{(\gamma_l - w)} - 1\bigr)}
       {\bigl(e^{(\gamma_l-w)} - e^{-w}\bigr)\bigl(-D(f_l \mid\mid g_l)\bigr)} +  O(1),
\end{align}
where, $O(1)$ again denotes terms that are bounded independently of $\gamma_l$ or $w$~\cite{tartakovsky2014sequential,gut2009stopped}. In a similar way, we can denote the post-change expected stopping time starting from $W_l(0)=w$ as $\mathbb{E}_{0}^l\!\bigl(T_l^{(w)}\bigr)$. Such partial-cycle expectations arise if we begin an SPRT cycle from a non-zero state $w$.

Integrating over all possible initial states yields the full expected stopping time $\mathbb{E}_{\infty}(T_l)$ under no change:

\scriptsize
\begin{align*}
&\mathbb{E}_{\infty}(T_l) \\
&= 1 + \int_{0}^{\gamma_l} 
      \mathbb{E}_{\infty}\bigl(T_l^{(w)}\bigr)
      \; d\mathbb{P}_{\infty}(Z_l=w)\\
&= 1 + \int_{0}^{\gamma_l}
    \frac{
      (\gamma_l - w)(1-e^{-\,w}) - w\bigl(e^{(\gamma_l-w)} - 1\bigr)
    }{\bigl(e^{(\gamma_l-w)}-e^{-w}\bigr)\bigl(-D(f_l \mid\mid g_l)\bigr)}\; d\mathbb{P}_{\infty}(Z_l=w) + O(1).
\end{align*}
\normalsize
We add $1$ to account for the first observation.

\begin{figure*}[t!]
    \centering
    \begin{eqnarray}
    \arl_l &\geq&\frac{\frac{1-(u(\gamma_l))^n}{1-u(\gamma_l)} + \frac{1-(u(\gamma_{l'}))^n}{1-u(\gamma_{l'})}(u(\gamma_l))^n + \tau(u(\gamma_l))^n(1+(u(\gamma_{l'}))^n)}{1-(u(\gamma_l))(u(\gamma_{l'}))^n} +O(1)
    \label{eq:lower_bound}
    \end{eqnarray}
    \begin{eqnarray}
    \arl_l &\leq& \frac{(1+C_l)\frac{1-(v(\gamma_l))^n}{1-v(\gamma_l)} + (1+C_{l'}) \frac{1-(v(\gamma_{l'}))^n}{1-v(\gamma_{l'})}(v(\gamma_l))^n + \tau(v(\gamma_l))^n(1+(v(\gamma_{l'}))^n)}{1-(v(\gamma_l))^n (v(\gamma_{l'})^n)} +O(1)
    \label{eq:upper_bound}
    \end{eqnarray}
    \hrule
\end{figure*}

The following lemma provides simplified upper bounds on $\mathbb{E}_{\infty}\!\bigl(T_l^{(w)}\bigr)$ and $\mathbb{E}_{\infty}(T_l)$ under no change, which will be used in ARL and WADD analyses.

\begin{lemma}\label{lemma:bound_E0_tau_l_w_AND_E0_tau_l}
The quantities $\mathbb{E}_{\infty}(T_{l}^{(w)})$ and $\mathbb{E}_{\infty}(T_{l})$ are upper-bounded as:
\begin{align*}
\mathbb{E}_{\infty}\bigl(T_{l}^{(w)}\bigr) 
&\le \frac{w}{D(f_l\mid \mid g_l)} + O(1),\\
\mathbb{E}_{\infty}\bigl(T_{l}\bigr)
&\le 1
 +\frac{\mathbb{E}_{\infty}\!\bigl(Z_{l}\,\mathbf{1}_{\{Z_{l}\ge0\}}\bigr)}{D(f_l \mid \mid g_l)}
 +O(1),
\end{align*}
where $Z_{l}=Z_l(t)=\log \Bigl(\tfrac{g_l(Y_{l,t})}{\,f_l(Y_{l,t})}\Bigr)$.
\end{lemma}
\begin{proof}
See Appendix~\ref{appendix:proof-of-lemma-bound_E0_tau_l_w_AND_E0_tau_l}.
\end{proof}
Before using these results for ARL or WADD bounds, we define the following ladder variables\cite{siegmund1985sequential}:
\begin{align*}
    T_l^- &= \min \{t \geq 1: W_l(t) \leq 0 \mid W_l(0) = 0\},  \\
    T_l^+ &= \min \{t \geq 1: W_l(t) > 0 \mid W_l(0) = 0\}.
\end{align*}
 We also define a stopping time
 \begin{align*}
     T_l^{\gamma_l} &= \min \{t \geq 1: W_l(t) \geq \gamma_l \mid W_l(0) = 0\}.
 \end{align*}

Under mild assumptions (finite second moments and strictly positive mean of log-likelihood ratio), we have that $\mathbb{E}_{\infty}(T_l^-)$ and $\mathbb{E}_0(T_l^+)$ are finite constants depending only on $f_l(\cdot)$ and $g_l(\cdot)$; see~\cite{lorden1970excess,kiefer1963asymptotically,siegmund1985sequential}. In particular, note that $T_l^+$ is geometrically distributed with parameter $q_l = \mathbb{P}_0^l(Z_l>0)$.

We now establish a few auxiliary lemmas to assist in analyzing the WADD and ARL bounds in subsequent subsection. 
\begin{lemma}\label{lemma:bound_beta}
The probability $\beta_l$ satisfies the following upper bound
\[
  \beta_l \leq 1 - q_l,
\]
where $q_l=\mathbb{P}_0^l(Z_l >0)$. 
\end{lemma}
\begin{proof}
See Appendix~\ref{appendix:proof-of-lemma-bound_beta}. A similar statement appears in \cite[Lemma~1, Eq.~(36)]{xu2021optimum}.
\end{proof}
\begin{lemma}\label{lemma:Lower-bound-psi_0l_w}
A lower bound on $\psi_{\infty,l}^{(w)}$ is
\[
  \psi_{\infty,l}^{(w)} 
  \;\ge\;
  \tfrac12\,\exp\!\bigl[-D\bigl(f_l\mid\mid g_l\bigr)\bigr]
  -
  \exp\bigl(-w\bigr).
\]
\end{lemma}
\begin{proof}
See Appendix~\ref{appendix:proof-of-lemma-Lower-bound-psi_0l_w}.
\end{proof}

The lower bound on $\psi_{\infty,l}^{(w)}$ is negative for $w=0$ and small $w$. In our application (see \eqref{eqn:bound_E0_tau^w}), we are primarily concerned with the behavior for \(w>0\). For our purposes, this bound is sufficient to establish the desired asymptotic behavior in the WADD analysis.  In fact, when \(w=0\), the expected remaining sojourn time $\mathbb{E}_{\infty} (\tau_{l'}^{(w)})$ becomes  exactly \(\mathbb{E}_{\infty}\bigl(\tau_{l'}\bigr)\), which is handled separately in the analysis. 

\begin{lemma}\label{lemma:Lower-bound-E_0(T_l)}
An upper bound on the expected stopping time of one round of SPRT in \eqref{eqn:stopping_time} post-change, denoted by $\mathbb{E}_0^l\bigl(T_l\bigr)$, is given by
\[
  \mathbb{E}_0^l\bigl(T_l\bigr)\leq
  \frac{\,q_l\,\gamma_l}{D\bigl(g_l\mid\mid f_l\bigr)}+O(1).
\]
\end{lemma}
\begin{proof}
See Appendix~\ref{appendix:proof-of-lemma-Lower-bound-E_0(T_l)}.
\end{proof}

\subsection{ARL Analysis of LS-CD}
\label{subsection:ARL_analysis}

We now analyze the ARL2FA for each location \(l\) under the proposed LS-CD algorithm. Define \(\theta_l \neq 0\) such that
\[
  \mathbb{E}_{\infty}\bigl(e^{\theta_l Z_l}\bigr) = 1.
\]
Since the mean drift under \(\mathbb{P}_{\infty}\) is negative, i.e.\ \(\mathbb{E}_{\infty} (Z_l)<0\), it follows that \(\theta_l>0\) \cite{ross1995stochastic}. 
In particular, $\theta_l=1$ is a convenient choice that we use throughout this paper.

Recall from Subsection~\ref{subsection:one-sided-sprt} that
\(\psi_{\infty,l}=\mathbb{P}_{\infty}\bigl(W_l(T_l)\ge\gamma_l\bigr)\) 
and 
\(\mathbb{E}_{\infty}(T_l)\) 
capture the probability of a false alarm and the expected stopping time, respectively, for a single SPRT cycle at location~\(l\). Under LS-CD, the UAV may run multiple SPRT cycles at location $l$, leading to repeated opportunities for a false alarm.

We then write the ARL2FA at location \(l\) as:
\[
  \arl_l
  \;=\;
  \sum_{k=1}^{n_l}
    (\bar{\psi}_{\infty,l})^{k-1}\,\mathbb{E}_{\infty}(T_l)
  \;+\;
  (\bar{\psi}_{\infty,l})^{n_l}\,\bigl(\tau + \arl_{l'}\bigr),
\]
where, \(\bar{\psi}_{\infty,l} \doteq 1-\psi_{\infty,l}\). These two equations can be solved for \( \arl_A \) and \( \arl_B \). The sum models repeated SPRT cycles where no false alarm is detected in the first \( k-1 \) cycles but a false alarm occurs in the \( k \)-th cycle. Intuitively, with probability \((\bar{\psi}_{\infty,l})^{k-1}\) there is no false alarm in the first \(k-1\) cycles, and $\mathbb{E}_{\infty}(T_l)$ is the expected time of an SPRT cycle under no change. If no false alarm occurs in $n_l$ consecutive cycles, the UAV switches to $l'$, incurring travel time $\tau$ and an ARL2FA of \(\arl_{l'}\). 

We now provide explicit lower and upper bounds on $\arl_l$ in the following proposition.

\begin{theorem}\label{theorem:ARL bound-n>1}
The ARL2FA at location $l$ for LS-CD algorithm with $n_A=n_B=n$ satisfies the lower and upper bounds given in \eqref{eq:lower_bound} and \eqref{eq:upper_bound}, where $u(\gamma_l) = 1-e^{-\gamma_l}$, $C_l = 1+\frac{\mathbb{E}_{\infty}(Z_{l} \mathbb{I}_{\{Z_{l} \geq 0\}})}{D(f_l||g_l)}$, and $v(\gamma_l)=1-K_l e^{-\gamma_l}$, $K_l=q_l \exp\Big\{-\frac{ J_{l}}{q_l D(g_l||f_l)}\Big\}$ and $J_{l} = \int \left(\log \frac{g_l(Y_{l,t})}{f_l(Y_{l,t})} \right)^2 g_l(Y_{l,t}) dy_{l,t}$. 
\end{theorem}
\begin{proof}
See appendix \ref{appendix:proof-of-arl-bound-n>1}.
\end{proof}

\begin{corollary}\label{corollary:special-arl-bound}
For the special case when the locations are symmetrical, i.e., $\gamma_A=\gamma_{B}=\gamma$, $n_A = n_B = n$, $f_{A}(\cdot)=f_{B}(\cdot)=f(\cdot)$ and $g_{A}(\cdot)=g_{B}(\cdot)=g(\cdot)$, we have,
\begin{eqnarray*}
    \liminf_{\gamma \to \infty}\frac{\arl_l}{e^{\gamma}} \geq 1+\frac{\tau}{n}
\end{eqnarray*}
and
\begin{eqnarray*}
\limsup_{\gamma \to \infty} \frac{\arl_l}{e^{\gamma}} \leq  \frac{1}{K}\bigg(1+C+\frac{\tau}{n}\bigg) 
\end{eqnarray*}
$\forall l \in \mathcal{L}$. Here, $K$ and $C$ are constants independent of the threshold as defined in Theorem \ref{theorem:ARL bound-n>1}. In the case of a single-location setup (i.e., $\tau=0$), the asymptotic result reduces to
\begin{align*}
    e^{\gamma}\leq \arl_l \leq  K' e^{\gamma}
\end{align*}
where $K' =\frac{1+C}{K}$ is a constant independent of the threshold.
\end{corollary}
\begin{proof}
See appendix \ref{appendix:proof-of-corollary}.
\end{proof}

\begin{figure*}[t!]
\centering
\begin{align} \label{eqn:temp_wadd}
   \wadd_l &= \max \Biggl\{ \underbrace{2 \tau+ \mathbb{E}_{\infty} (\tau_{l'})}_{S_1} \ , \ \underbrace{\tau+ \sup_{w \in (0,\gamma_{l'})} \mathbb{E}_{\infty} \bigl(\tau_{l'}^{(w)} \bigr)}_{S_2} \ , \nonumber\\
   &\qquad \qquad \underbrace{\sup_{\substack{w \in (0,\gamma_{l}) \\ 1 \leq m \leq n_l}} \Bigl[\mathbb{E}_0^l \bigl(T_{l}^{(w)}\bigr) + \beta_l^{(w)}\sum_{k=m+1}^{n_l} (\beta_l)^{k-(m+1)} \mathbb{E}_0^l(T_l) +\beta_l^{(w)}(\beta_l)^{n_l-m} \bigl(2 \tau+ \mathbb{E}_{\infty} (\tau_{l'}) \bigr)  \Bigr]}_{S_3} \Biggr\} \nonumber\\
   &\quad \ +\sum_{k=1}^{n_l} (\beta_l)^{k-1} \mathbb{E}_0^l(T_l)+(\beta_l)^{n_l} \widetilde{\wadd}_l 
\end{align}
\hrule
\end{figure*}

\begin{remark}
The lower and upper bounds presented in Corollary~\ref{corollary:special-arl-bound} for the ARL exhibit explicit dependence on the UAV's operational parameters, such as the travel time $\tau$ and the switching count $n$. Importantly, for the limiting case where $\tau = 0$ and $n = 1$ corresponding to a scenario without movement delays or switching constraints, the bounds reduce to those of standard CUSUM applied to a single process. While the lower bound to ARL2FA of standard CUSUM is well-studied, explicit upper bounds for the ARL2FA have not been previously reported in the literature. Our results address this gap by providing an explicit upper bound which is only a constant factor away from the lower bound.
\end{remark}

\subsection{WADD Analysis of LS-CD}
\label{subsection:WADD-analysis}

Next, we turn to the worst-case average detection delay (WADD) when a change occurs at location $l$ at time $\nu$. Let $\wadd_l$ denote this worst-case detection delay. Define $\tilde{\theta}_l\neq 0$ such that
\[
  \mathbb{E}_0^l\!\bigl(e^{\,\tilde{\theta}_l\,Z_l}\bigr)=1.
\]
Since the drift under $\mathbb{P}_0^l$ is positive (i.e.\ $\mathbb{E}_0^l(Z_l) > 0$), it follows that $\tilde{\theta}_l < 0$; see~\cite{ross1995stochastic}. In particular, $\tilde{\theta}_l=-1$ is a convenient choice that we use throughout this paper.

$\wadd_l$ is achieved if the UAV is at one of the following three states at time $\nu$:
\begin{enumerate}[label=\textit{(\alph*)}]
\item The UAV departs from $l$ exactly at time $\nu$.
\item The UAV is stationed at $l'$ at time $\nu$ in its first SPRT cycle, with partial statistic $W_{l'}(\nu)=w\in(0,\gamma_{l'})$.
\item The UAV is already at $l$ in the middle of the $m^{\text{th}}$ SPRT $(1\leq m\leq n_l)$ with partial statistic $W_l(\nu)=w>0$. We later show (Lemma~\ref{lemma:stochastic_arg}) that this scenario yields a \emph{stochastically smaller} delay than scenario (a) and hence can be omitted.
\end{enumerate}

Recall that the expected completion time of the ongoing SPRT at location $l'$ is denoted by $\mathbb{E}_{\infty}\big(T_{l'}^{(w)}\big)$ if $W_{l'}(\nu)=w$, and the probability that this random walk at $l'$ raises an alarm is $\psi_{\infty,l'}^{(w)}$. Now, in scenario (b), if at time~$\nu$ the UAV is at~$l'$ with $W_{l'}(\nu)=w$, then the mean \emph{remaining} sojourn time at $l'$ can be written as:
\begin{equation}\label{eqn:mean_rem_soj_time}
  \mathbb{E}_{\infty} \bigl(\tau_{l'}^{(w)}\bigr) =
  \mathbb{E}_{\infty}\bigl(T_{l'}^{(w)}\bigr)+
  \bar{\psi}_{\infty,l'}^{(w)}
  \sum_{k=1}^{\,n_{l'}-1}
    \Bigl(\bar{\psi}_{\infty,l'}\Bigr)^{k-1}\,
    \mathbb{E}_{\infty}\!\bigl(T_{l'}\bigr)
\end{equation}
On the other hand, the mean sojourn time at $l'$ under no change is
\begin{equation}\label{eqn:mean_soj_time}
  \mathbb{E}_{\infty} \bigl(\tau_{l'}\bigr)=
  \sum_{k=1}^{\,n_{l'}}
    \bigl(\bar{\psi}_{\infty,l'}\bigr)^{k-1}
    \mathbb{E}_{\infty} \!\bigl(T_{l'}\bigr).
\end{equation}
Now, suppose that the UAV fails to detect a change at location $l$ after $n_l$ consecutive attempts (i.e., the local statistic $W_l(t)$ returns to zero $n_l$ times without hitting $\gamma_l$). At that point, the UAV switches to location $l'$, thereby resetting the detection statistic at $l$. We denote by $\widetilde{\wadd}_l$ the additional mean detection delay under this full reset. Formally,
\begin{equation}\label{eqn:tilde_wadd}
  \widetilde{\wadd}_l =
  2\tau +  \mathbb{E}_{\infty}\bigl(\tau_{l'}\bigr) + 
  \sum_{k=1}^{n_l} \bigl(\beta_l\bigr)^{k-1}
    \mathbb{E}_0^l\bigl(T_l\bigr) +
  \bigl(\beta_l\bigr)^{n_l}\widetilde{\wadd}_l.
\end{equation}
In this expression, $2\tau$ is the round-trip travel overhead, $\mathbb{E}_{\infty}(\tau_{l'})$ is the pre-change expected sojourn time at $l'$, and $\mathbb{E}_0^l\bigl(T_l\bigr)$ is the expected stopping time at location $l$ post-change.

Combining the scenarios above, the worst-case average detection delay at location~$l$ can be decomposed as shown in \eqref{eqn:temp_wadd}.

\subsubsection*{Interpreting the Terms in \eqref{eqn:temp_wadd}}
\begin{itemize}
\item \textbf{Scenario (a):} 
  $S_1 = 2\tau + \mathbb{E}_{\infty}(\tau_{l'})$ arises if the UAV has \emph{just left} $l$ when the change occurs (time $\nu$). It must effectively travel back and forth, incurring a $2\tau$ overhead plus a full pre-change sojourn time at $l'$.

\item \textbf{Scenario (b):}
  $S_2 = \tau + \sup_{w \in (0,\gamma_{l'})}\,\mathbb{E}_{\infty}(\tau_{l'}^{(w)})$ represents the case the UAV is already at $l'$ with partial statistic $w\in(0,\gamma_{l'})$ and in its first SPRT. It completes the partial sojourn at $l'$, then travels $\tau$ time slots to $l$. The worst-case is captured by taking the maximum over all $w \in (0,\gamma_{l'})$.

\item \textbf{Scenario (c):} $S_3$ captures the possibility of being already at $l$ in the $m^{\text{th}}$ SPRT with partial statistic $w$. However, it can be shown to produce a stochastically smaller delay compared to (a) (see Lemma~\ref{lemma:stochastic_arg}). We still include it here in \eqref{eqn:temp_wadd}  for completeness, but the worst-case typically arises from scenario (a) or (b).

\item Finally, $\sum_{k=1}^{n_l}(\beta_l)^{k-1}\mathbb{E}_0^l(T_l)$ corresponds to repeated post-change SPRT attempts that eventually succeed in raising an alarm,  and $(\beta_l)^{n_l}\,\widetilde{\wadd}_l$ accounts for the full reset if the UAV switches away after $n_l$ consecutive SPRT failures, as described in \eqref{eqn:tilde_wadd}.
\end{itemize}

It is important to note that the expression for $S_2$ involves optimization over $w \in (0,\gamma_{l'})$. Obviously, if $w=0$, then the mean detection delay in Scenario (b) is exactly $S_1+ \sum_{k=1}^{n_l} (\beta_l)^{k-1} \mathbb{E}_0^l(T_l)+(\beta_l)^{n_l} \widetilde{\wadd}_l -\tau$, since the sojourn at location $l'$ is just starting in this case. On the other hand, if $w=\gamma_{l'}$, then the mean detection delay in Scenario (b) is only $\sum_{k=1}^{n_l} (\beta_l)^{k-1} \mathbb{E}_0^l(T_l)+(\beta_l)^{n_l} \widetilde{\wadd}_l + \tau$ where the UAV has just finished its sojourn at location $l'$. Since Scenario (a) produces higher mean detection delay than Scenario (b) for both $w=0$ and $w=\gamma_{l'}$, we can safely ignore $w=0$ and $w=\gamma_{l'}$ and retain the optimization under Scenario (b) only over $w \in (0,\gamma_{l'})$ (in the expression of $S_2$). The maximizer over $[0,\gamma_{l'}]$,  $w^*$, may or may not exist within $(0,\gamma_{l'})$ in this case. It is also noted that if $w^* \in \{0,\gamma_{l'}\}$, we will have $S_1>S_2$.   Similar logic has been applied to the expression of $S_3.$

The following Lemma establishes the dominance of $S_1$ over $S_3$ in \eqref{eqn:temp_wadd}.
\begin{lemma}\label{lemma:stochastic_arg}
\[S_3 \leq S_1.\]
\end{lemma}
\begin{proof}
See Appendix~\ref{appendix:proof-of-lemma-stochastic-arg}.
\end{proof}

In summary, regardless of the partial index $m$ or the value of $w > 0$, beginning detection at location~$l$ immediately upon the change is stochastically faster than either leaving $l$ at time $\nu$ or being at $l'$ with $w \in (0, \gamma_{l'})$. This ensures that $\wadd_l$ is not enlarged beyond the delays captured by scenarios (a) or (b) in \eqref{eqn:temp_wadd} and thus the modified $\wadd_l$ can be written as:
\begin{align} \label{eqn:wadd}
   \wadd_l &= \max \Bigl\{2 \tau+ \mathbb{E}_{\infty} (\tau_{l'}), \ \tau+ \sup_{w \in (0,\gamma_{l'})} \mathbb{E}_{\infty} \bigl(\tau_{l'}^{(w)} \bigr)\Bigr\} \nonumber\\
   &\ + \sum_{k=1}^{n_l} (\beta_l)^{k-1} \mathbb{E}_0^l(T_l)+(\beta_l)^{n_l} \widetilde{\wadd}_l
\end{align}

The following theorem establishes an upper bound on $\wadd_l$.

\begin{theorem}\label{theorem:WADD bound}
The worst-case average detection delay for the LS-CD algorithm satisfies the following bound under the assumption that $n_A=n_B=n$:
\begin{eqnarray*}
    \wadd_l \leq \frac{\gamma_l }{D(g_l||f_l)} + C' + O(1)
\end{eqnarray*}
\(\forall l \in \mathcal{L} \). 
Here, $C'$ is independent of the threshold $\gamma_l$ and is given by \( C' = \max\bigg\{C_1+\frac{\gamma_{l'}}{D(f_{l'}||g_{l'})},C_2\bigg\} \).
The constants $C_1$ and $C_2$ are defined as follows:

\footnotesize
\begin{align*}
    C_1 &=  \tau\bigg(\frac{1+\bar{q}_l^n}{1-\bar{q}_l^n}\bigg)+\left(1+\frac{\mathbb{E}_{\infty}(Z_{l'} \mathbbm{1}_{\{Z_{l'} \geq 0\}})}{D(f_{l'}||g_{l'})}+O(1)\right)\\
    &\qquad \qquad \qquad \quad \times \bigg(\frac{n\bar{q}_l^n}{1-\bar{q}_l^n} + (n-1)(1+e^{-w^*}-\frac{1}{2}e^{-D(f_{l'}||g_{l'})}\bigg)
\end{align*}
\normalsize
and 
\begin{align*}
C_2 = \frac{2 \tau+ n +\frac{n\mathbb{E}_{\infty}(Z_{l'} \mathbbm{1}_{\{Z_{l'} \geq 0\}})}{D(f_{l'}||g_{l'})}+O(1)}{1-\bar{q}_l^n}
\end{align*}
Here $\bar{q_l} = 1-q_l$ and $w^\ast = \argmax_{w \in (0,\gamma_{l'})}\mathbb{E}_{\infty} \bigl(\tau_{l'}^{(w)} \bigr)$ is assumed to exist. If $w^* \notin (0,\gamma_{l'})$, then 
\begin{equation*}
\wadd_l \leq \frac{\gamma_l}{D(g_l||f_l)} + C_2 + O(1).
\end{equation*}

\end{theorem}
\begin{proof}
See appendix \ref{appendix:proof-of-theorem_1}. 
\end{proof}
\begin{remark}
    The $e^{-w^*}$ term in the expression of $C_1$ can be trivially upper bounded by $1$.
\end{remark}

The WADD bound in Theorem~\ref{theorem:WADD bound} retains the hallmark features of classical CUSUM: it scales linearly with the threshold~\(\gamma_l\) and inversely with the Kullback--Leibler divergence~\(D(g_{l}\mid\mid f_{l})\). An additive term~\(C'\), independent of~\(\gamma_l\), reflects the UAV's operational overhead, such as travel time~\(\tau\) and sojourn cycles at the other location~\(l'\). In particular, if the UAV sojourns at location~\(l'\) (with threshold~\(\gamma_{l'}\)) while the actual change occurs at~\(l\), it must first confirm the absence of change at~\(l'\). Even when no change is actually present at $l'$, this process incurs a detection delay of approximately~\(\tfrac{\gamma_{l'}}{\,D(f_{l'} \mid\mid g_{l'})} +O(1)\).

Moreover, in the special case when \(n=1\), \(\tau=0\) and \(\gamma_{l'}=0\) (i.e.\ a single-location setup with uninterrupted observations), our result reduces to the well-known linear upper bound for WADD of standard CUSUM.

\begin{remark}
    Though LS-CD has been designed for two locations, the algorithm and its performance analyses can easily be extended to multiple locations under a round-robin switching mechanism for location changing.
\end{remark}

\section{Numerical Results}\label{sec:numerical_results}
In this section, we present simulation results evaluating the trade-offs among detection delay, thresholds, and energy consumption for the \emph{LS-CD} algorithm.

\subsection{Setup and Parameter Choices}

We consider the observation distributions $f_l = \mathcal{N}(0,1)$ and $g_l = \mathcal{N}(2,1)$. The simulation parameters are as follows:
\begin{itemize}
  \item \textbf{Sensing energy} $E_s=1$ unit per time step,
  \item \textbf{Movement energy} $E_m=4$ units per time step,
  \item \textbf{Travel time} $\tau=3$ steps between locations,
  \item \textbf{Energy threshold} $\bar{E}=3$,
  \item \textbf{ARL constraints} $\mathrm{r}_A= 500,\;\mathrm{r}_B = 500$.
\end{itemize}
For each parameter tuple $(\gamma_A, \gamma_B, n)$ in LS-CD, we compute the corresponding $\arl_A$, $\arl_B$, and the average energy consumption using simulations under the no-change scenario. The worst-case average detection delay ($\wadd$) is computed using the LS-CD formula. We also verify feasibility by checking compliance with all constraints: energy consumption, $\arl_A$, and $\arl_B$.

\begin{figure}[ht!]
    \centering
    \begin{subfigure}[t]{\linewidth}
        \centering
        \includegraphics[width=0.75\linewidth]{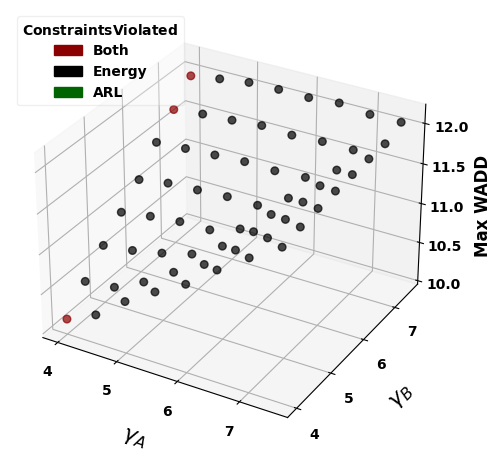}
        \caption{$n=1$}
        \label{fig:plot_n1}
    \end{subfigure}
    \vspace{2mm}
    \begin{subfigure}[t]{\linewidth}
        \centering
        \includegraphics[width=\linewidth]{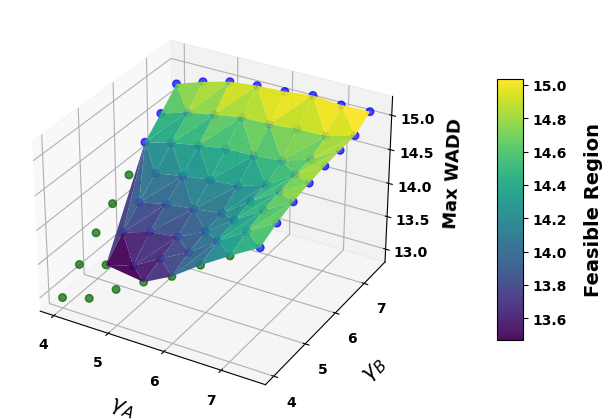}
        \caption{$n=3$}
        \label{fig:plot_n3}
    \end{subfigure}
    \vspace{2mm}
    \begin{subfigure}[t]{\linewidth}
        \centering
        \includegraphics[width=\linewidth]{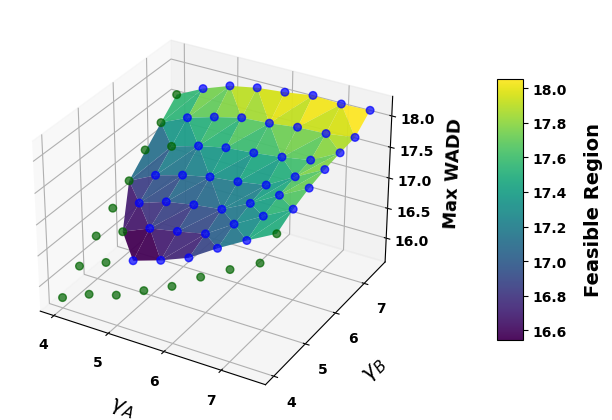}
        \caption{$n=5$}
        \label{fig:plot_n5}
    \end{subfigure}
    \caption{3D plots of $\max\{\wadd_A,\wadd_B\}$ vs.\ $(\gamma_A,\gamma_B)$ for LS-CD ($n=1,3,5$). Points are color-coded by whether they violate energy only (black), ARL (green), both (red), or are feasible (blue). A triangulated surface is drawn over feasible points.}
    \label{fig:surface-Plot}
\end{figure}
\vspace{-3mm}

\subsection{Results}
\begin{figure}[ht!]
    \centering
    \includegraphics[width=0.95\linewidth]{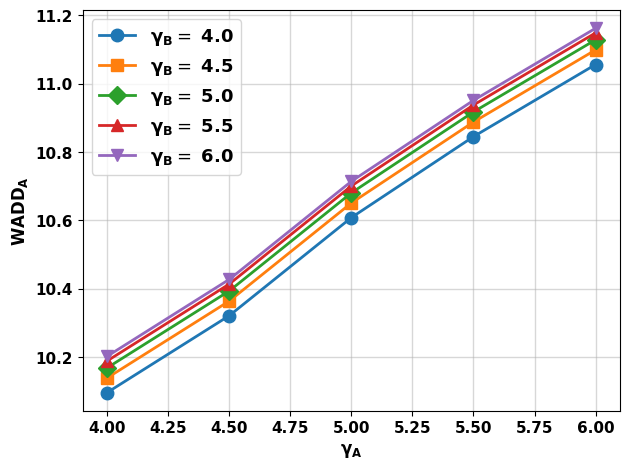}
    \vspace{-2mm}
    \caption{$\wadd_A$ vs.\ $\gamma_A$ for different $\gamma_B$ under LS-CD with $n=1$.}
    \label{fig:wadd_n1_Plot}
    \vspace{-2mm}
\end{figure}

\begin{figure}[t!]
    \centering
    \includegraphics[width=0.95\linewidth]{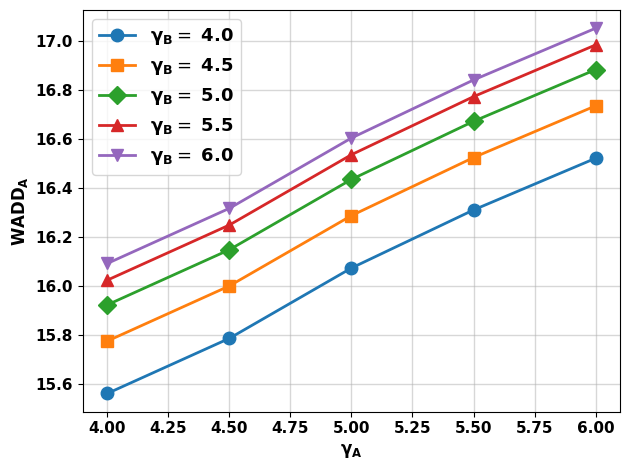}
    \vspace{-2mm}
    \caption{$\wadd_A$ vs.\ $\gamma_A$ for different $\gamma_B$ under LS-CD with $n=5$.}
    \label{fig:wadd_n5_Plot}
    \vspace{-2mm}
\end{figure}

Figure~\ref{fig:surface-Plot} demonstrates the objective function of our optimization problem (\ref{opt:CP}), $\max\{\wadd_A,\wadd_B\}$, across a grid of thresholds $(\gamma_A,\gamma_B)$ for $n\in \{1,3,5\}$. Each point is color-coded according to constraint satisfaction or violation. A triangulated surface delineates the feasible region—where all constraints in (\ref{opt:CP}) are satisfied. As $n$ increases, the energy consumption of the UAV decreases at the expense of higher detection delays. Specifically, we observe a favorable feasible region when $n=3$, as it effectively balances   WADD and constraint satisfaction.

Figures~\ref{fig:wadd_n1_Plot} and \ref{fig:wadd_n5_Plot} illustrate how the LS-CD worst-case detection delay at location~$A$, denoted $\wadd_A$, varies with threshold $\gamma_A$ for different values of $\gamma_B$ at $n=1$ and $n=5$, respectively. As expected, $\wadd_A$ increases monotonically with $\gamma_A$ because more evidence is needed to declare a change. Larger $\gamma_B$ also elevates $\wadd_A$, since the UAV can sojourn longer at $B$ while missing opportunities to detect at $A$. With $n=5$, the UAV remains longer at location $B$ under scenarios (a) and (b) of \eqref{eqn:temp_wadd}, further increasing $\wadd_A$ compared to the case when $n=1$.

\begin{figure}[ht!]
    \centering
    \includegraphics[width=0.95\linewidth]{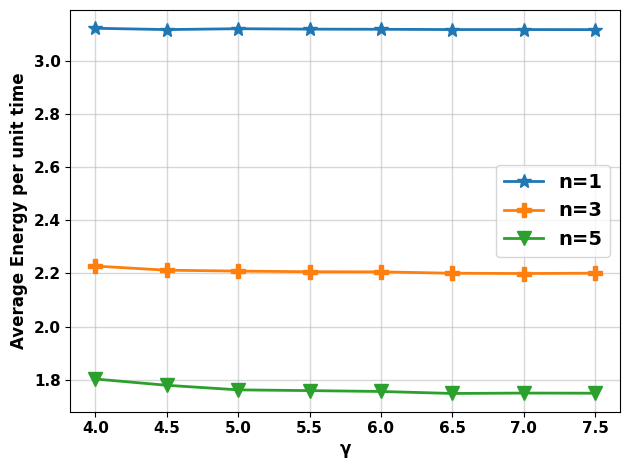}
    \vspace{-2mm}
    \caption{Average energy consumption vs.\ the common threshold $\gamma=\gamma_A=\gamma_B$ for~LS-CD with $n\in\{1,3,5\}$}.
    \label{fig:energy_all_Plot}
    \vspace{-2mm}
\end{figure}

In Figure~\ref{fig:energy_all_Plot}, we plot the average energy per unit time against a common threshold value, $\gamma = \gamma_A = \gamma_B$, comparing scenarios for $n\in \{1,3,5\}$. LS-CD with $n=1$ leads to frequent travel, thus consuming higher energy. Conversely, $n=5$ can curb movement cost significantly and, thus, energy expenditure. 

Although increasing $n$ reduces energy use, the marginal improvement lessens with larger $n$ values, highlighting the trade-off between detection delay and energy efficiency. Lower $n$ values tend to violate energy constraints, while higher $n$ values ensure compliance with the energy budget at the expense of longer detection delays. Thus, there exists a delicate balance between location switching frequency, threshold selection, and energy constraints. Our results indicate that LS-CD with a moderate switching threshold (e.g., $n=3$) provides near-optimal performance in many practical scenarios.

\section{Conclusion}\label{sec:conclusion}
In this paper, we have proposed the LS-CD algorithm for QCD in two locations monitored by a single UAV, with the goal of reducing WADD while meeting  constraints on ARL2FA and energy consumption rate of the UAV.  Theoretical bounds on ARL2FA and WADD of LS-CD were derived, revealing a number of interesting trade-offs. Interestingly, we have derived a new asymptotic upper bound to the ARL2FA of the standard CUSUM algorithm. However, the work can be extended to handle multiple challenges such as (i) multiple locations under surveillance by a single UAV, (ii) providing robustness against  uncertainties in UAV travel times and observation distributions, (iii) unknown pre and post change distributions,  and  (iv)  handling correlated observation across  different locations. We plan to address these challenges in our future research endeavours.

\appendices 

 \section{Proof of lemma \ref{lemma:bound_E0_tau_l_w_AND_E0_tau_l}}
\label{appendix:proof-of-lemma-bound_E0_tau_l_w_AND_E0_tau_l}
We have,
\small
\begin{align}\label{eqn:bound_expected stopping time_w}
\mathbb{E}_{\infty}({T_{l}^{(w)}}) &=\frac{(\gamma_{l}-w)(1-e^{- w})-w(e^{(\gamma_{l}-w)}-1)}{(e^{(\gamma_{l}-w)}-e^{- w}) (-D(f_l||g_l))} + O(1)\nonumber\\
&= \frac{(\gamma_{l}-w)(e^{ w}-1)-w(e^{\gamma_{l}}-e^{ w})}{(e^{ \gamma_{l}}-1) (-D(f_l||g_l))} + O(1)\nonumber\\
&=\frac{\gamma_{l}(e^{ w}-1)}{(e^{ \gamma_{l}}-1)(-D(f_l||g_l))} + \frac{w}{D(f_l||g_l)} + O(1)\nonumber\\
&\leq \frac{w}{D(f_l||g_l)}+O(1)
\end{align}
\normalsize
Also, we have, 
\begin{align}\label{eqn:bound_expected stopping time}
\mathbb{E}_{\infty}(T_{l})&= 1+\int_{w'=0}^{\gamma_{l}} \mathbb{E}_{\infty}(T_{l}^{(w')}) d \mathbb{P}_{\infty}(Z_{l}=w') \nonumber\\
&\leq 1+\int_{w'=0}^{\gamma_{l}} \frac{w'}{D(f_l||g_l)} d \mathbb{P}_{\infty}(Z_{l}=w')  + O(1) \nonumber\\
&\leq1+ \frac{1}{D(f_l||g_l)}\int_{w'=0}^{\infty} w' d \mathbb{P}_{\infty}(Z_{l}=w')  + O(1) \nonumber\\
&=1+\frac{\mathbb{E}_{\infty}(Z_{l} \mathbbm{1}_{\{Z_{l} \geq 0\}})}{D(f_l||g_l)}  + O(1)
\end{align}
where, $Z_{l}=Z_l(t)=\log \left(\frac{g_l(Y_{l,t})}{f_l(Y_{l,t})}\right)$ is the log-likelihood ratio. Equations \eqref{eqn:bound_expected stopping time_w} and \eqref{eqn:bound_expected stopping time} prove the lemma.

\section{Proof of lemma \ref{lemma:bound_beta}}
\label{appendix:proof-of-lemma-bound_beta}

Note that $\{W_l(T_l) \leq 0\} \implies \{T_l^- < \infty\}$. Also, $\{W_l(T_l) \leq 0\} \subseteq \{T_l^- < \infty\}$ \cite{xu2021optimum}.
 Hence,
\begin{eqnarray*}
    \mathbb{P}_0^l(W_l(T_l) \leq 0) \leq \mathbb{P}_0^l(T_l^- < \infty)
\end{eqnarray*}
We also have the following equality from \cite[Corollary 8.39]{siegmund1985sequential},
\begin{eqnarray*}
    \mathbb{E}_0^l(T_l^+) = \frac{1}{\mathbb{P}_0^l(T_l^- = \infty)}
\end{eqnarray*}
Therefore,
\begin{eqnarray*}
    \frac{1}{1-\beta_l} &=&\frac{1}{1-\mathbb{P}_0^l(W_l(T_l) \leq 0)}\\
    &\leq&\frac{1}{1-\mathbb{P}_0^l(T_l^- < \infty)}\\
    &=&\frac{1}{\mathbb{P}_0^l(T_l^- = \infty)}=\mathbb{E}_0^l(T_l^+)=\frac{1}{q_l}
\end{eqnarray*}
Using the above inequality, the proof is complete.

\section{Proof of lemma \ref{lemma:Lower-bound-psi_0l_w}}
\label{appendix:proof-of-lemma-Lower-bound-psi_0l_w}
Using the Bretagnolle-Huber inequality \cite[Chapter 14]{lattimore2020bandit}, we have:
\begin{align*}
    \psi_{\infty,l}^{(w)} + \bar{\psi}_{0,l}^{(w)} \geq \frac{1}{2}e^{-D(f_l||g_l)}
\end{align*}
Following the derivation of the bound $\psi_{\infty,l}^{(w)} \leq e^{-(\gamma_{l}-w)}$ from \cite{ross1995stochastic}, we can also derive a similar bound for  $\bar{\psi}_{0,l}^{(w)}$. Specifically, we have:
\[\bar{\psi}_{0,l}^{(w)} \leq e^{-w},\] 
Combining these results, we obtain:
\begin{align*}
    \psi_{\infty,l}^{(w)} \geq \frac{1}{2}e^{-D(f_l||g_l)} - e^{-w}
\end{align*}

\section{Proof of lemma \ref{lemma:Lower-bound-E_0(T_l)}}
\label{appendix:proof-of-lemma-Lower-bound-E_0(T_l)}
Using a parallel expression like that of $\mathbb{E}_{\infty}(T_l^{(w)})$ in \eqref{eqn:E_infty_T_l_w} and then the bound as presented in Lemma~\ref{lemma:bound_E0_tau_l_w_AND_E0_tau_l}, we can write:
\begin{align}\label{eqn:bound_E_0_T_l_w}
\mathbb{E}_0^l(T_l^{(w)}) \leq \frac{\gamma_l - w}{D(g_l||f_l)} + O(1)
\end{align}
Now, we write:
\begin{align*}
    &\mathbb{E}_0^l(T_l) \\
    &= 1+\int_{w=0}^{\gamma_l} \mathbb{E}_0^l(T_l^{(w)}) d\mathbb{P}_0^l(Z_l=w)  \\
    & \leq 1+ \int_{0}^{\gamma_l} \frac{\gamma_l - w}{D(g_l||f_l)} d\mathbb{P}_0^l(Z_l=w)+O(1) \\
    &= 1+ \frac{\mathbb{E}_0^l\big((\gamma_l - Z_l) \mathbbm{1}_{\{0 < Z_l \leq \gamma_l\}}\big)}{D(g_l||f_l)}+O(1)\\
    &=1+\mathbb{P}_0^l(0 < Z_l \leq \gamma_l) \frac{\mathbb{E}_0^l\big((\gamma_l - Z_l) \mid 0 < Z_l \leq \gamma_l\big)}{D(g_l||f_l)} +O(1)\\
    &\leq 1+\frac{\gamma_l}{D(g_l||f_l)} \mathbb{P}_0^l(
    0 < Z_l \leq \gamma_l)+O(1)\\
    &\leq 1+\frac{\gamma_l}{D(g_l||f_l)} \mathbb{P}_0^l(Z_l > 0)+O(1)\\
    &= \frac{\gamma_l}{D(g_l||f_l)}\cdot q_l +O(1)
\end{align*}

\section{Proof of Theorem \ref{theorem:ARL bound-n>1}}
\label{appendix:proof-of-arl-bound-n>1}
The average run length to false alarm at location $l \in \mathcal{L}$ is given by \eqref{eqn:ARL-eqn-n>1} when $n_A = n_B = n$.
\begin{figure*}[ht!]
\centering
\begin{eqnarray}\label{eqn:ARL-eqn-n>1}
    \arl_l &=&\sum_{k=1}^{n_l} (\bar{\psi}_{\infty,l})^{k-1} \mathbb{E}_{\infty}(T_l)+(\bar{\psi}_{\infty,l})^{n_l} (\tau+\arl_{l'}) \nonumber\\
    &=&\sum_{k=1}^{n} (\bar{\psi}_{\infty,l})^{k-1} \mathbb{E}_{\infty}(T_l)+(\bar{\psi}_{\infty,l})^{n} \bigg(\tau +\sum_{k=1}^{n} (\bar{\psi}_{\infty,l'})^{k-1} \mathbb{E}_{\infty}(T_{l'})+(\bar{\psi}_{\infty,l'})^{n} (\tau+\arl_{l})\bigg) \nonumber\\
    \implies \arl_l&=&\frac{\mathbb{E}_{\infty}(T_l)\sum\limits_{k=1}^{n} (\bar{\psi}_{\infty,l})^{k-1} + \mathbb{E}_{\infty}(T_{l'})(\bar{\psi}_{\infty,l})^{n}\sum\limits_{k=1}^{n} (\bar{\psi}_{\infty,l'})^{k-1} + \tau (\bar{\psi}_{\infty,l})^{n}(1+(\bar{\psi}_{\infty,l'})^{n})}{1-(\bar{\psi}_{\infty,l}\bar{\psi}_{\infty,l'})^{n}}\\
    \text{where $n_A=n_B=n$.} \nonumber
\end{eqnarray}
\label{fig:ARL_eqn_appendix}
\hrule
\end{figure*}

We first establish a bound on \(\psi_{\infty,l}\). Using results from \cite[Chapter 7]{ross1995stochastic}, we have:
\begin{align}\label{eqn:psi_bound}
    \mathbb{E}_{\infty}\big(e^{W_l(T_l)} \big) 
    &= \psi_{\infty,l} \mathbb{E}_{\infty}\big(e^{W_l(T_l)} \mid W_l(T_l) \geq \gamma_l \big) \nonumber\\
    &\quad + \bar{\psi}_{\infty,l}   \mathbb{E}_{\infty}\big(e^{W_l(T_l)} \mid W_l(T_l) \leq 0 \big) \nonumber\\
    &\geq \psi_{\infty,l}  \mathbb{E}_{\infty}\big(e^{W_l(T_l)} \mid W_l(T_l) \geq \gamma_l \big) \nonumber\\
    & \geq  \psi_{\infty,l} \ e^{\gamma_l}
\end{align}
We use the Martingale stopping theorem\cite{ross1995stochastic} to obtain \(\mathbb{E}_{\infty}\big(e^{ W_l(T_l)} \big) = 1\). Hence, \eqref{eqn:psi_bound} becomes:
\begin{align}
    \psi_{\infty,l} \cdot e^{\gamma_l} &\leq 1 \nonumber\\
    \psi_{\infty,l}  &\leq  e^{-\gamma_l} \nonumber\\
    \implies \bar{\psi}_{\infty,l} &\geq 1 - e^{-\gamma_l} \doteq u(\gamma_l).
\end{align}
We also find the following bound:
\begin{align*}
    \sum\limits_{k=1}^{n} (\bar{\psi}_{\infty,l})^{k-1}\geq  \sum\limits_{k=1}^{n} (u(\gamma_l))^{k-1} =\frac{1-(u(\gamma_l))^n}{1-u(\gamma_l)}
\end{align*}

Now, using all this information and the fact that $\mathbb{E}_{\infty}(T_l)\geq 1$ in \eqref{eqn:ARL-eqn-n>1}, we have:
\scriptsize
\begin{align}\label{eqn:ARL-lower-bound-n>1}
&\arl_l \nonumber\\
&\geq \frac{\frac{1-(u(\gamma_l))^n}{1-u(\gamma_l)} + \frac{1-(u(\gamma_{l'}))^n}{1-u(\gamma_{l'})}(u(\gamma_l))^n + \tau (u(\gamma_l))^n(1+(u(\gamma_{l'}))^n)}{1-(u(\gamma_l))^n(u(\gamma_{l'}))^n} +O(1)
\end{align}
\normalsize
Also,
\small
\begin{align*}
    \psi_{\infty,l}&= \mathbb{P}_{\infty}(W_l(T_l)\geq \gamma_l) \\
    &= \mathbb{E}_{\infty}[\mathbbm{1}_{\{W_l(T_l) \geq \gamma_l\}}]\\
    &\overset{(a)}{=}\mathbb{E}_0^l[e^{-W_l(T_l)}\mathbbm{1}_{\{W_l(T_l) \geq \gamma_l\}}]\\
    &=e^{-\gamma_l}\mathbb{E}_0^l[e^{-(W_l(T_l)- \gamma_l)} \mid W_l(T_l) \geq \gamma_l] (1-\beta_l)\\
    &\overset{(b)}{\geq} e^{-\gamma_l} \exp\left\{-\mathbb{E}_0^l[(W_l(T_l)- \gamma_l) \mid W_l(T_l) \geq \gamma_l]\right\} (1-\beta_l) \\
    &=e^{-\gamma_l} \exp\left\{-\frac{\mathbb{E}_0^l[(W_l(T_l)- \gamma_l) \mathbbm{1}_{\{W_l(T_l) \geq \gamma_l\}}]}{1-\beta_l}\right\} (1-\beta_l)\\
    &\overset{(c)}{\geq} e^{-\gamma_l} \exp\left\{-\frac{\mathbb{E}_0^l[W_l(T_l)- \gamma_l]}{1-\beta_l}\right\} (1-\beta_l) \\
    &\overset{(d)}{\geq} e^{-\gamma_l} e^{-\frac{ J_{l}}{q_l D(g_l||f_l)}}\cdot q_l  \\
    &= K_l e^{-\gamma_l}\\
    \implies \bar{\psi}_{\infty,l} &\leq 1- K_l e^{-\gamma_l} \doteq v(\gamma_l)
\end{align*}
\normalsize
where we use the change of measure argument in (a) and Jensen's inequality in (b). Note that for inequality (c), we can write $T_l$ defined in \eqref{eqn:stopping_time} in terms of ladder variables defined in Section~\ref{subsection:one-sided-sprt} as $T_l=\min\{T_l^{\gamma_l}, T_l^{-}\}$. Thus, we have \((W_l(T_l)- \gamma_l) \mathbbm{1}_{\{W_l(T_l) \geq \gamma_l\}} = (W_l(T_l)- \gamma_l) \mathbbm{1}_{\{T_l^{\gamma_l} < T_l^{-}\}} \leq (W_l(T_l)- \gamma_l) \). We also use Lemma \ref{lemma:bound_beta} to bound $\beta_l$. In inequality (d), we use a bound on the expected overshoot as given in \cite[Corollary 1]{lorden1970excess} to get $\mathbb{E}_0^l[W_l(T_l)-\gamma_l] \leq \frac{J_{l}}{D(g_l||f_l)}$ and $J_{l} = \int \left(\log \frac{g_l(Y_{l,t})}{f_l(Y_{l,t})} \right)^2 g_l(Y_{l,t}) dy_{l,t}$. Here, $K_l=q_l \exp\Big\{-\frac{ J_{l}}{q_l D(g_l||f_l)}\Big\}$. 

Using Lemma \ref{lemma:bound_E0_tau_l_w_AND_E0_tau_l}, we also have:
\begin{eqnarray*}
    \mathbb{E}_{\infty}(T_{l}) &\leq& 1+\frac{\mathbb{E}_{\infty}
    (Z_{l} \mathbb{I}_{\{Z_{l} \geq 0\}})}{D(f_l||g_l)} +O(1)\\
    &=&1+C_l+O(1)
\end{eqnarray*}
where, $C_l=\frac{\mathbb{E}_{\infty}(Z_{l} \mathbb{I}_{\{Z_{l} \geq 0\}})}{D(f_l||g_l)}$ is a constant independent of $\gamma_l$.

We also find the following bound:
\begin{eqnarray*}
    \sum\limits_{k=1}^{n} (\bar{\psi}_{\infty,l})^{k-1}&\leq& \sum\limits_{k=1}^{n} (v(\gamma_l))^{k-1}\\
    &=&\frac{1-(v(\gamma_l))^n}{1-v(\gamma_l)}
\end{eqnarray*}
where $v(\gamma_l)=1-K_l e^{-\gamma_l}$.

Therefore, using all these information in \eqref{eqn:ARL-eqn-n>1}, we have:
\tiny
\begin{align}\label{eqn:ARL-upper-bound-n>1}
    &\arl_l \nonumber\\
    &\leq \frac{(1+C_l)\frac{1-(v(\gamma_l))^n}{1-v(\gamma_l)} + (1+C_{l'}) \frac{1-(v(\gamma_{l'}))^n}{1-v(\gamma_{l'})}(v(\gamma_l))^n + \tau(v(\gamma_l))^n(1+(v(\gamma_{l'}))^n)}{1-(v(\gamma_l))^n (v(\gamma_{l'}))^n} \nonumber \\
    &\qquad + O(1)
\end{align}
\normalsize
This proves the theorem.

\section{Proof of corollary \ref{corollary:special-arl-bound}}
\label{appendix:proof-of-corollary}
In the symmetric case when $\gamma_l=\gamma_{l'}=\gamma$, \eqref{eqn:ARL-lower-bound-n>1} becomes:
\begin{align}\label{eqn:special-arl}
\arl_l &\geq \frac{\frac{1-(u(\gamma))^n}{1-u(\gamma)}(1 + (u(\gamma))^n) + \tau(u(\gamma))^n(1+(u(\gamma))^n)}{1-(u(\gamma))^{2n}} \nonumber\\
&=\frac{1}{1-u(\gamma)} + \frac{\tau (u(\gamma))^n}{1-(u(\gamma))^{n}} \nonumber\\
&=\frac{1}{1-(1-e^{-\gamma})} + \frac{\tau (1-e^{-\gamma})^n}{1-(1-e^{-\gamma})^{n}}  \nonumber\\
&=e^{\gamma} + \frac{\tau (1-e^{-\gamma})^n}{1-(1-e^{-\gamma})^{n}}
\end{align}
Now, let,
\begin{align*}
    h_1(\gamma)&=\frac{(1-e^{-\gamma})^n}{1-(1-e^{-\gamma})^{n}}\\
    &=\frac{1-ne^{-\gamma} + {\binom{n}{2}}e^{-2\gamma} +\ldots+(-1)^ne^{-n\gamma}}{ne^{- \gamma} - {\binom{n}{2}}e^{-2\gamma} +\ldots+(-1)^n e^{-n \gamma}}
    \end{align*}
Clearly,
\begin{eqnarray*}
    \lim_{\gamma \to \infty} \frac{h_1(\gamma)}{e^{\gamma}}
    &=&\frac{1}{n}
\end{eqnarray*}
Hence, 
\begin{eqnarray}\label{eqn:special-arl-lower-bound}
    \liminf_{\gamma \to \infty} \frac{\arl_l}{e^{\gamma}} \geq  1+\frac{\tau}{n}
\end{eqnarray}
Now, again for the symmetric case, we have from \eqref{eqn:ARL-upper-bound-n>1}:
\begin{align*}
&\arl_l \\
&\leq \frac{(1+C)\frac{1-(v(\gamma))^n}{1-v(\gamma)}(1+(v(\gamma))^n) + \tau(v(\gamma))^n(1+(v(\gamma))^n)}{1-(v(\gamma))^{2n}}\\
&=\frac{(1+C)}{1-v(\gamma)} + \frac{\tau(v(\gamma))^n}{1-(v(\gamma))^{n}}\\
&= \frac{(1+C)}{K e^{-\gamma}} + \frac{\tau (1-K e^{-\gamma})^n}{1-(1-K e^{-\gamma})^n}\\
&=  \frac{(1+C)e^{\gamma}}{K} + \frac{\tau (1-K e^{-\gamma})^n}{1-(1-K e^{-\gamma})^n}
\end{align*}
Now let,
\begin{eqnarray*}
    h_2(\gamma) &=& \frac{ (1-K e^{-\gamma})^n}{1-(1-K e^{-\gamma})^n} \\
    &=& \frac{1-nK e^{-\gamma} + \ldots + (-1)^n K^n e^{-n\gamma}}{nK e^{-\gamma} + \ldots + (-1)^n K^n e^{-n\gamma}}
\end{eqnarray*}
Clearly,
\begin{eqnarray*}
    \lim_{\gamma \to \infty} \frac{h_2(\gamma)}{e^{ \gamma}}
    &=&\frac{1}{nK}
\end{eqnarray*}
Hence, 
\begin{eqnarray}\label{eqn:special-arl-upper-bound}
    \limsup_{\gamma \to \infty} \frac{\arl_l}{e^{\gamma}} \leq  \frac{1}{K}\bigg(1+C + \frac{\tau}{n}\bigg) 
\end{eqnarray}

\section{Proof of Lemma~\ref{lemma:stochastic_arg}}
\label{appendix:proof-of-lemma-stochastic-arg}

\begin{figure*}[ht!]
\begin{align}\label{eqn:case1_bound}
    \wadd_l &\leq \tau \frac{1+\bar{q}_l^n}{1-\bar{q}_l^n}+\frac{\bar{q}_l^n}{1-\bar{q}_l^n} \left( n+\frac{n\mathbb{E}_{\infty}(Z_{l'} \mathbbm{1}_{\{Z_{l'} \geq 0\}})}{D(f_{l'}||g_{l'})}+O(1)\right)+\frac{\gamma_l \ q_l }{q_l D(g_l||f_l)}+O(1)+\frac{w^*}{D(f_{l'}||g_{l'})}+O(1) \nonumber\\
    &\quad +(n-1) \biggl( 1+ e^{-w^*}-\frac{1}{2}e^{-D(f_{l'}||g_{l'})} \biggr)\left(1+\frac{\mathbb{E}_{\infty}(Z_{l'} \mathbbm{1}_{\{Z_{l'} \geq 0\}})}{D(f_{l'}||g_{l'})} +O(1)\right)\nonumber\\
    &\leq\frac{\gamma_l}{D(g_l||f_l)} + \frac{\gamma_{l'}}{D(f_{l'}||g_{l'})} + C_1 +O(1)
\end{align}
where we use the bound on $\mathbb{E}_0^l(T_l)$ from Lemma~\ref{lemma:Lower-bound-E_0(T_l)}. Also, note that \(w^\ast \leq \gamma_{l'}\). Here,
\begin{align*}
    C_1 &=  \tau\bigg(\frac{1+\bar{q}_l^n}{1-\bar{q}_l^n}\bigg)+\left(1+\frac{\mathbb{E}_{\infty}(Z_{l'} \mathbbm{1}_{\{Z_{l'} \geq 0\}})}{D(f_{l'}||g_{l'})}+O(1)\right)\times \bigg(\frac{n\bar{q}_l^n}{1-\bar{q}_l^n} + (n-1)(1+e^{-w^*}-\frac{1}{2}e^{-D(f_{l'}||g_{l'})}\bigg)
\end{align*}
is independent of the threshold $\gamma_l$. Note that $w^\ast$ can at most be $\gamma_{l'}$, but we can upper bound $e^{-w^\ast}$ by $1$. \\
\hrule
\end{figure*}

We show that the random delay incurred due to scenario (c)---the UAV \emph{already} at $l$ in the middle of its $m^{\text{th}}$ SPRT, partial statistic $w>0$---is stochastically smaller than that due to scenario (a) (see \eqref{eqn:temp_wadd}). Stochastic ordering theory says that  a random variable $U$ is   stochastically smaller than $V$ (i.e., $U \preceq V$)   if $\mathbb{P}(U>t)\le \mathbb{P}(V>t)$ for all $t \in \mathbb{R}$.  

Let $D_{m,w}$ denote the random detection  delay starting from the $m^{\text{th}}$ SPRT cycle at location $l$ with $W_l(\nu)=w$. Let $D'$ denote the random detection delay if the UAV has just left location $l$ at time $\nu$. 

Obviously, $D_{m,w} \preceq D_{m,0}$ since, post change, starting from $w>0$ will require stochastically smaller time to raise an alarm compared to starting from $w=0$ in the same SPRT cycle. Also, $D_{m,0} \preceq D'$ since starting the $m$-th SPRT cycle at location $l$ at time $\nu$ provides an immediate opportunity to the UAV to raise an alarm at location $l$. Hence, 
\[D_{m,w} \preceq D'\]
and consequently $\mathbb{E}_0^l(D_{m,w}) \leq \mathbb{E}_0^l(D') \ \forall \ 1 \leq m \leq n, w \in (0,\gamma_l)$. Thus, $\sup_{1 \leq m \leq n,w \in (0,\gamma_l)} \mathbb{E}_0^l(D_{m,w}) \leq \mathbb{E}_0^l(D')$.

This shows that the worst-case delay when the UAV is already at location $l$ (scenario (c)) is no greater than the delay incurred when the UAV has just left $l$ (scenario (a)). In our notation, this means that $S_3 \leq S_1$.

\section{Proof of theorem \ref{theorem:WADD bound}}
\label{appendix:proof-of-theorem_1}
From \eqref{eqn:wadd} and \eqref{eqn:tilde_wadd} and using $n_l=n$  $\forall l \in \mathcal{L}$, we have:
\footnotesize
\begin{align*}   
   &\wadd_l \\
   &=\max\{2 \tau+ \mathbb{E}_{\infty} (\tau_{l'}), \tau+ \sup_{w \in (0,\gamma_{l'})} \mathbb{E}_{\infty} (\tau_{l'}^{(w)}) \}  + \sum_{k=1}^{n} (\beta_l)^{k-1} \mathbb{E}_0^l(T_l)\\
   &\quad +(\beta_l)^{n}\left( \frac{2 \tau+ \mathbb{E}_{\infty} (\tau_{l'})+ \sum_{k=1}^{n} (\beta_l)^{k-1} \mathbb{E}_0^l(T_l)}{1-(\beta_l)^{n}}\right)\\
   &=\max\{S_1,S_2\} + \sum_{k=1}^{n} (\beta_l)^{k-1} \mathbb{E}_0^l(T_l) \\
   &\quad + (\beta_l)^{n}\left( \frac{2 \tau+ \mathbb{E}_{\infty} (\tau_{l'})+ \sum_{k=1}^{n} (\beta_l)^{k-1} \mathbb{E}_0^l(T_l)}{1-(\beta_l)^{n}}\right)
\end{align*}
\normalsize
Now we consider two cases as follows:\\

$\textsc{Case } 1: \text{If } S_2 > S_1$,
then we have,
\footnotesize
\begin{align} \label{eqn:WADD_case_1} 
&\wadd_l \nonumber\\
&=  \frac{\tau (1 + (\beta_l)^n)}{1 - (\beta_l)^n} +  \frac{\mathbb{E}_{\infty}(\tau_{l'}) (\beta_l)^n}{1 - (\beta_l)^n}  + \frac{\mathbb{E}_0^l(T_l) \sum\limits_{k=1}^{n} (\beta_l)^{k-1}}{1 - (\beta_l)^n} + \mathbb{E}_{\infty}(\tau_{l'}^{(w^*)})
\end{align}
\normalsize

Note that $\frac{\sum\limits_{k=1}^{n} (\beta_l)^{k-1}}{1 - (\beta_l)^n} = \frac{1}{1 - \beta_l}$.

We now find a bound on each term. We know,
\begin{equation*}
 \mathbb{E}_{\infty} (\tau_{l'}) =  \sum_{k=1}^{n} (\bar{\psi}_{\infty,l'})^{k-1} \mathbb{E}_{\infty}(T_{l'}) 
\end{equation*}

Using $\sum_{k=1}^{n} (\bar{\psi}_{\infty,l'})^{k-1} \leq n$ and using Lemma \ref{lemma:bound_E0_tau_l_w_AND_E0_tau_l}, we get a bound on $\mathbb{E}_{\infty} (\tau_{l'})$ as:
\begin{eqnarray}\label{eqn:bound_E0_tau}
    \mathbb{E}_{\infty} (\tau_{l'}) \leq n+\frac{n \mathbb{E}_{\infty}(Z_{l'} \mathbbm{1}_{\{Z_{l'} \geq 0\}})}{D(f_{l'}||g_{l'})} +O(1)
\end{eqnarray}

Further, we have:
\begin{equation*}
 \mathbb{E}_{\infty} (\tau_{l'}^{(w)}) = \mathbb{E}_{\infty}(T_{l'}^{(w)})  + \bar{\psi}_{\infty,l'}^{(w)} \sum_{k=1}^{n-1} (\bar{\psi}_{\infty,l'})^{k-1} \mathbb{E}_{\infty}(T_{l'}) 
\end{equation*}
Using Lemma \ref{lemma:bound_E0_tau_l_w_AND_E0_tau_l}, we have a bound on $\mathbb{E}_{\infty}(T_{l'}^{(w)})$ and $\mathbb{E}_{\infty}(T_{l'})$.
Now, by Lemma \ref{lemma:Lower-bound-psi_0l_w},
\begin{eqnarray*}
    \bar{\psi}_{\infty,l'}^{(w)} &\leq& 1+ e^{-w}-\frac{1}{2}e^{-D(f_{l'}||g_{l'})} 
\end{eqnarray*}
Also, $\sum_{k=1}^{n-1} (\bar{\psi}_{\infty,l'})^{k-1} \leq n-1$. 

\noindent Therefore,

\scriptsize
\begin{align}\label{eqn:bound_E0_tau^w}
&\mathbb{E}_{\infty} (\tau_{l'}^{(w)}) \nonumber\\
&\leq \frac{w}{D(f_{l'}||g_{l'})} +O(1)\nonumber\\
&+(n-1)\biggl( 1+ e^{-w}-\frac{1}{2}e^{-D(f_{l'}||g_{l'})} \biggr)\left(1+\frac{\mathbb{E}_{\infty}(Z_{l'} \mathbbm{1}_{\{Z_{l'} \geq 0\}})}{D(f_{l'}||g_{l'})}+O(1)\right) 
\end{align}
\normalsize
Using Lemma \ref{lemma:bound_beta}, we get, $1-(\beta_l)^n\geq 1-(1-q_l)^n$.

Also let, $\bar{q}_l=1-q_l$.
Using Lemma \ref{lemma:bound_beta} and the bounds from \eqref{eqn:bound_E0_tau} and \eqref{eqn:bound_E0_tau^w} in \eqref{eqn:WADD_case_1}, we get the upper bound as shown in \eqref{eqn:case1_bound}.

$\textsc{Case } 2:  \text{If } S_2<S_1$,
then,
\begin{eqnarray}\label{eqn:WADD_case_2}
   \wadd_l = \frac{ 2 \tau+ \mathbb{E}_{\infty}(\tau_{l'})+\mathbb{E}_0^l(T_l)\sum\limits_{k=1}^{n} (\beta_l)^{k-1}}{1-(\beta_l)^{n}}
\end{eqnarray}

Using \eqref{eqn:bound_E0_tau}, Lemma \ref{lemma:bound_beta} and the bound on $\mathbb{E}_0^l(T_l)$ in \eqref{eqn:WADD_case_2}, we have,
\scriptsize
\begin{align}\label{eqn:case2_bound}
    &\wadd_l \nonumber\\
    &\leq \frac{ 2 \tau+ n+\frac{n\mathbb{E}_{\infty}(Z_{l'} \mathbbm{1}_{\{Z_{l'} \geq 0\}})}{D(f_{l'}||g_{l'})}+O(1)+\bigg(\frac{\gamma_l \ q_l}{q_l D(g_l||f_l)}+O(1)\bigg)\left(1-\bar{q}_l^n\right)}{ 1-\bar{q}_l^n} \nonumber\\
    &=\frac{\gamma_l}{D(g_l||f_l)} + C_2 + O(1)
\end{align}
\normalsize
where,
\begin{align*}
    C_2 = \frac{2 \tau+ n +\frac{n\mathbb{E}_{\infty}(Z_{l'} \mathbbm{1}_{\{Z_{l'} \geq 0\}})}{D(f_{l'}||g_{l'})}+O(1)}{1-\bar{q}_l^n}
\end{align*}
is a constant independent of the threshold.

Therefore \eqref{eqn:case1_bound} and \eqref{eqn:case2_bound} give the bounds on $\wadd_l$ at the two locations. Finally combining case $1$ and case $2$, we have:
\begin{align} 
   \wadd_l &\leq \frac{\gamma_l}{D(g_l||f_l)} + C'+O(1)
\end{align}
where, 
\begin{align*}
    C' = \max\{C_1 + \frac{\gamma_{l'}}{D(f_{l'}||g_{l'})},C_2\}
\end{align*}
which gives the final bound and proves the theorem. 

If $w^* \notin (0,\gamma_{l'})$, then $S_1>S_2$ and hence $\wadd_l \leq \frac{\gamma_l}{D(g_l||f_l)} + C_2 + O(1).$

\bibliographystyle{unsrt}
\bibliography{reference.bib}

\begin{thebibliography}{10}

\bibitem{poor2008quickest}
H~Vincent Poor and Olympia Hadjiliadis.
\newblock {\em Quickest detection}.
\newblock Cambridge University Press, 2008.

\bibitem{shiryaev1963optimum}
Albert~N Shiryaev.
\newblock On optimum methods in quickest detection problems.
\newblock {\em Theory of Probability \& Its Applications}, 8(1):22--46, 1963.

\bibitem{lorden1971procedures}
Gary Lorden.
\newblock Procedures for reacting to a change in distribution.
\newblock {\em The annals of mathematical statistics}, pages 1897--1908, 1971.

\bibitem{moustakides1986optimal}
George~V Moustakides.
\newblock Optimal stopping times for detecting changes in distributions.
\newblock {\em the Annals of Statistics}, 14(4):1379--1387, 1986.

\bibitem{page1954continuous}
Ewan~S Page.
\newblock Continuous inspection schemes.
\newblock {\em Biometrika}, 41(1/2):100--115, 1954.

\bibitem{lai1998information}
Tze~Leung Lai.
\newblock Information bounds and quick detection of parameter changes in stochastic systems.
\newblock {\em IEEE Transactions on Information theory}, 44(7):2917--2929, 1998.

\bibitem{lai2010sequential}
Tze~Leung Lai and Haipeng Xing.
\newblock Sequential change-point detection when the pre-and post-change parameters are unknown.
\newblock {\em Sequential analysis}, 29(2):162--175, 2010.

\bibitem{veeravalli2001decentralized}
Venugopal~V Veeravalli.
\newblock Decentralized quickest change detection.
\newblock {\em IEEE Transactions on Information theory}, 47(4):1657--1665, 2001.

\bibitem{tartakovsky2004change}
Alexander~G Tartakovsky and Venugopal~V Veeravalli.
\newblock Change-point detection in multichannel and distributed systems with applications.
\newblock {\em STATISTICS TEXTBOOKS AND MONOGRAPHS}, 173:339--370, 2004.

\bibitem{mei2010efficient}
Yajun Mei.
\newblock Efficient scalable schemes for monitoring a large number of data streams.
\newblock {\em Biometrika}, 97(2):419--433, 2010.

\bibitem{xie2013sequential}
Yao Xie and David Siegmund.
\newblock Sequential multi-sensor change-point detection.
\newblock In {\em 2013 Information theory and applications workshop (ITA)}, pages 1--20. IEEE, 2013.

\bibitem{fellouris2017multistream}
Georgios Fellouris, George~V Moustakides, and Venu~V Veeravalli.
\newblock Multistream quickest change detection: Asymptotic optimality under a sparse signal.
\newblock In {\em 2017 IEEE International Conference on Acoustics, Speech and Signal Processing (ICASSP)}, pages 6444--6447. IEEE, 2017.

\bibitem{tartakovsky2014sequential}
Alexander Tartakovsky, Igor Nikiforov, and Michele Basseville.
\newblock {\em Sequential analysis: Hypothesis testing and changepoint detection}.
\newblock CRC press, 2014.

\bibitem{banerjee2015data}
Taposh Banerjee and Venugopal~V Veeravalli.
\newblock Data-efficient minimax quickest change detection in a decentralized system.
\newblock {\em Sequential Analysis}, 34(2):148--170, 2015.

\bibitem{rovatsos2021quickest}
Georgios Rovatsos, George~V Moustakides, and Venugopal~V Veeravalli.
\newblock Quickest detection of moving anomalies in sensor networks.
\newblock {\em IEEE Journal on Selected Areas in Information Theory}, 2(2):762--773, 2021.

\bibitem{wang2015large}
Yuan Wang and Yajun Mei.
\newblock Large-scale multi-stream quickest change detection via shrinkage post-change estimation.
\newblock {\em IEEE Transactions on Information Theory}, 61(12):6926--6938, 2015.

\bibitem{xie2019asynchronous}
Liyan Xie, Yao Xie, and George~V Moustakides.
\newblock Asynchronous multi-sensor change-point detection for seismic tremors.
\newblock In {\em 2019 IEEE International Symposium on Information Theory (ISIT)}, pages 787--791. IEEE, 2019.

\bibitem{raghavan2010quickest}
Vasanthan Raghavan and Venugopal~V Veeravalli.
\newblock Quickest change detection of a markov process across a sensor array.
\newblock {\em IEEE Transactions on Information Theory}, 56(4):1961--1981, 2010.

\bibitem{zou2019quickest}
Shaofeng Zou, Venugopal~V Veeravalli, Jian Li, and Don Towsley.
\newblock Quickest detection of dynamic events in networks.
\newblock {\em IEEE Transactions on Information Theory}, 66(4):2280--2295, 2019.

\bibitem{hadjiliadis2009one}
Olympia Hadjiliadis, Hongzhong Zhang, and H~Vincent Poor.
\newblock One shot schemes for decentralized quickest change detection.
\newblock {\em IEEE Transactions on Information Theory}, 55(7):3346--3359, 2009.

\bibitem{xu2021optimum}
Qunzhi Xu, Yajun Mei, and George~V Moustakides.
\newblock Optimum multi-stream sequential change-point detection with sampling control.
\newblock {\em IEEE Transactions on Information Theory}, 67(11):7627--7636, 2021.

\bibitem{zhang2019partially}
Chen Zhang and Steven~CH Hoi.
\newblock Partially observable multi-sensor sequential change detection: A combinatorial multi-armed bandit approach.
\newblock In {\em Proceedings of the AAAI Conference on Artificial Intelligence}, volume~33, pages 5733--5740, 2019.

\bibitem{ross1995stochastic}
Sheldon~M Ross.
\newblock {\em Stochastic processes}.
\newblock John Wiley \& Sons, 1995.

\bibitem{gut2009stopped}
Allan Gut.
\newblock {\em Stopped random walks}.
\newblock Springer, 2009.

\bibitem{siegmund1985sequential}
D.~Siegmund.
\newblock {\em Sequential Analysis: Tests and Confidence Intervals}.
\newblock Springer Series in Statistics. Springer, 1985.

\bibitem{lorden1970excess}
Gary Lorden.
\newblock On excess over the boundary.
\newblock {\em The Annals of Mathematical Statistics}, 41(2):520--527, 1970.

\bibitem{kiefer1963asymptotically}
Jack Kiefer and J~Sacks.
\newblock Asymptotically optimum sequential inference and design.
\newblock {\em The Annals of Mathematical Statistics}, pages 705--750, 1963.

\bibitem{lattimore2020bandit}
Tor Lattimore and Csaba Szepesv{\'a}ri.
\newblock {\em Bandit algorithms}.
\newblock Cambridge University Press, 2020.

\end{thebibliography}

\end{document}